\documentclass{article} 

\usepackage{microtype}
\usepackage{enumitem}
\usepackage[dvipsnames]{xcolor}
\usepackage{tikz}
\usepackage{algorithm,comment,bm}
\usepackage[noend]{algpseudocode}
\usepackage{amsmath,amsthm, amssymb, graphicx, url}
\usepackage[hidelinks]{hyperref}
\usepackage{cite}
\graphicspath{{figs/}}

\algnewcommand\algorithmicinput{\textbf{Initialization:}}
\algnewcommand\init{\item[\algorithmicinput]}
\algnewcommand\algorithmicevol{\textbf{Evolution:}}
\algnewcommand\evol{\item[\algorithmicevol]}

\renewcommand{\natural}{{\mathbb{N}}}
\newcommand{\real}{{\mathbb{R}}}

\newcommand{\until}[1]{\{1,\ldots,#1\}}

\newcommand{\EE}{\mathcal{E}}
\newcommand{\GG}{\mathcal{G}}

\newcommand{\NN}{\mathcal{N}}

\newcommand{\VV}{\mathcal{V}} 
\newcommand{\WW}{\mathcal{W}}

\newcommand{\E}{\mathbb{E}}

\newcommand{\m}{\mathop{\textrm{minimize}}}

\newcommand{\R}{\mathbb{R}}

\newcommand{\x}{\bm{x}}

\newcommand{\NNii}{N_{i}}

\newcommand{\avx}{\bar{x}}
\newcommand{\avxt}{\avx_t}
\newcommand{\avxtm}{\avx_{t-1}}
\newcommand{\avg}{\bar{g}}
\newcommand{\avgt}{\avg_t}
\newcommand{\avgtm}{\avg_{t-1}}
\newcommand{\avd}{\bar{d}}
\newcommand{\avdt}{\avd_t}
\newcommand{\avdtm}{\avd_{t-1}}

\newcommand{\bx}{\mathbf{x}}
\newcommand{\bxt}{\bx_t}
\newcommand{\bxtm}{\bx_{t-1}}
\newcommand{\bg}{\mathbf{g}}
\newcommand{\bgt}{\bg_t}
\newcommand{\bgtm}{\bg_{t-1}}
\newcommand{\btgt}{\tilde{\bg}_t}

\newcommand{\bd}{\mathbf{d}}
\newcommand{\bdt}{\bd_t}
\newcommand{\bdtm}{\bd_{t-1}}

\newcommand{\one}{\mathbf{1}}

\newcommand{\hP}{\hat{P}}
\newcommand{\hPit}{\hP_{i,t}}

\newcommand{\hq}{\hat{q}}
\newcommand{\hqit}{\hq_{i,t}}

\newcommand{\hr}{\hat{r}}
\newcommand{\hrit}{\hr_{i,t}}
\newcommand{\hU}{\hat{U}}
\newcommand{\hUit}{\hU_{i,t}}

\newcommand{\xstart}{x_\star(t)}
\newcommand{\xstartm}{x_\star(t-1)}
\newcommand{\fstart}{f_\star(t)}
\newcommand{\hx}{\hat{x}}
\newcommand{\hxstart}{\hx_\star(t)}
\newcommand{\hxstartm}{\hx_\star(t-1)}

\newcommand{\xit}{x_{i,t}}

\newcommand{\xitm}{x_{i,t-1}}
\newcommand{\xjtm}{x_{j,t-1}}
\newcommand{\dit}{d_{i,t}}

\newcommand{\ditm}{d_{i,t-1}}
\newcommand{\djtm}{d_{j,t-1}}
\newcommand{\git}{g_{i,t}}

\newcommand{\gitm}{g_{i,t-1}}

\newcommand{\yit}{y_{i,t}}

\newcommand{\eit}{\epsilon_{i,t}}

\newcommand{\xis}{x_{i,s}}
\newcommand{\yis}{y_{i,s}}

\newcommand{\hf}{\hat{f}}
\newcommand{\hfstart}{\hf_\star(t)}
\newcommand{\ta}{\xi}
\newcommand{\ha}{\hat{\ta}}
\newcommand{\tv}{\bm{\chi}}

\newcommand{\tai}{\ta_i}

\newcommand{\hait}{\ha_{i,t}}
\newcommand{\haitm}{\ha_{i,t-1}}

\newcommand{\tvis}{\bm{\chi}_{i,s}}
\newcommand{\tvit}{\bm{\chi}_{i,t}}

\newcommand{\Sit}{\mathbf{s}_{i,t}}
\newcommand{\Rit}{R_{i,t}}

\newcommand{\Ritm}{R_{i,t-1}}

\newcommand{\Vt}{\mathbf{v}_t}
\newcommand{\Vtm}{\mathbf{v}_{t-1}}
\newcommand{\Zt}{\mathbf{z}}
\newcommand{\Ztm}{\mathbf{z}}
\newcommand{\Ztau}{\mathbf{z}}

\newcommand{\Vtb}{\mathbf{v}_{\bar{t}}}

\newcommand{\bUB}{c_{\nabla} }
\newcommand{\cUB}{c_x}

\DeclareMathOperator{\col}{col}
\DeclareMathOperator{\unpack}{\textsc{unpack}}

\newcommand\oprocendsymbol{\hbox{$\square$}}
\newcommand\oprocend{\relax\ifmmode\else\unskip\hfill\fi\oprocendsymbol}
\def\eqoprocend{\tag*{$\square$}}

\newtheorem{theorem}{Theorem}[section]

 \newtheorem{lemma}[theorem]{Lemma}
\newtheorem{remark}[theorem]{Remark}

\newtheorem{assumption}[theorem]{Assumption}

\usepackage{authblk}

\author[1]{Ivano Notarnicola}
\author[2]{Andrea Simonetto}
\author[1]{Francesco Farina}
\author[1]{Giuseppe Notarstefano}

\affil[1]{\small Department of Electrical, Electronic and Information Engineering, 
	  University of Bologna, Bologna, Italy, \texttt{name.lastname@unibo.it}.}
\affil[2]{\small AI\&Quantum team at IBM Research Europe, Dublin, Ireland \texttt{andrea.simonetto@ibm.com}.}

\title{Distributed Personalized Gradient Tracking \\ with Convex Parametric Models	\thanks{
    This result is part of a project that has received funding from the European Research Council (ERC)
    under the European Union's Horizon 2020 research and innovation programme
    (grant agreement No 638992 - OPT4SMART). }
}
\usepackage{cite}

\begin{document}

\maketitle

\begin{abstract}%
We present a distributed optimization algorithm for solving online personalized optimization problems over a network of computing and communicating nodes, each of which linked to a specific user. The local objective functions are assumed to have a composite structure and to consist of a known time-varying (engineering) part and an unknown (user-specific) part. Regarding the unknown part, it is assumed to have a known parametric (e.g., quadratic) structure a priori, whose parameters are to be learned \emph{along with} the evolution of the algorithm. The algorithm is composed of two intertwined components: (i) a dynamic gradient tracking scheme for finding local solution estimates and (ii) a recursive least squares scheme for estimating the unknown parameters via user's noisy feedback on the local solution estimates. The algorithm is shown to exhibit a bounded regret under suitable assumptions. Finally, a numerical example corroborates the theoretical analysis.%
\end{abstract}

\section{Introduction}

Cyber-physical and social systems (CPSS) are becoming increasingly important in today's society, whenever human actions, preferences, and behaviors are added to the cyber and physical space~\cite{Dressler2018}. Important examples of this class of systems are the energy grid~\cite{Chatupromwong2012,Ospina2020}, transportation infrastructures~\cite{Quercia2014}, personalized healthcare~\cite{Menner2020}, and robotics~\cite{Luo2020}. 

A key feature of CPSS is the trade-off between given engineered performance metrics and user's (dis)comfort, perceived safety, and preferences. While, on one side, engineered goals may come from well-defined metrics based on physical models and can be time-varying to model data streams~\cite{dall2019optimization}, on the other side, user's (dis)satisfaction is more complex to model. The ``utility'' function to be optimized for the users is often based on averaged models constructed on generic one-fits-all models. However, good averaged models of users' utilities are difficult to obtain for the associated cost and time of human studies, the data is therefore scarce and biased. %
For these reasons, more tailored and personalized strategies are to be preferred when dealing with humans~\cite{simonetto_personalized_2019}. 

This paper studies time-varying optimization problems distributed across a network of $N$ agents. 
Each agent represents both a physical node (e.g., a home or a car) and its associated user. 
The optimization has a cost function that comprises of both a known time-varying 
engineering cost, and an unknown user specific (dis)satisfaction function. Formally, 
we define the \emph{distributed personalized problem} as
\begin{equation}\label{pb:problem}
\begin{aligned}
  & \m_{x \in \R^n}
  & & \sum_{i=1}^N \underbrace{V_i(x;t) + U_i (x)}_{f_i(x; t)}, & t \in \natural
\end{aligned}
\end{equation}
where $x \in\R^n$ represents the common decision variable, and $t\in \natural$ represents the time index; each agent $i$ is equipped 
with the known time-varying engineering cost $V_i(x;t): \R^n\times \natural \to \R$ and with the 
unknown user's dissatisfaction function $U_i(x):\R^n \to \R$. The aggregated cost $f_i(x,t):\R^n\times \natural\to\R$ 
is associated to agent $i$ only%
. Then, by \emph{solving} problem~\eqref{pb:problem}, we mean to 
generate a sequence of tentative solutions, say $\{\avxt\}_{t=1}^T$, which 
make the corresponding cost $\sum_{i=1}^N f_{i,t} (\avxt;t)$ as close as possible to its (current) 
optimal value, say $\fstart$, for all $t$. 
In particular, as customary in online optimization, we measure the quality of the given sequence $\{\avxt\}_{t=1}^T$ 
using the \emph{cumulative dynamic regret} up to time $T$ defined by
\begin{align}\label{eq:regret}
  R_T(\{\avxt\}_{t=1}^T) & \triangleq
  \sum_{t=1}^T \left ( \sum_{i=1}^N f_i(\avxt; t)  - \fstart \right )
\end{align}
and the \emph{average dynamic regret} up to time $T$ defined by $R_T(\{\avxt\}_{t=1}^T)/T$. 
As it is customary in the distributed setting, we also complement these measures 
with the consensus metric $C_{T}(\{x_{i,T}\}_{i=1}^N, \bar{x}_{T}) 
\triangleq \sum_{i=1}^N\|x_{i,T}-\bar{x}_{T}\|^2$, quantifying how far from consensus
the local decisions $x_{i,T}$ are at time $T$.

The challenges in solving problem~\eqref{pb:problem} are multiple. First, a \emph{distributed} strategy 
must be developed. 
Then, not only the optimization problem changes over time, but its cost function is not completely known 
by the agents and it has to be learned \emph{concurrently} to the solution of the problem, 
by employing noisy user's feedback. 

This paper addresses all the above mentioned challenges and provides the
following main contributions.

We propose a \emph{personalized gradient tracking} distributed scheme combining
an online optimization algorithm with a learning mechanism, and derive a bound
on its dynamic regret.
As a building block for the proposed scheme, we develop a dynamic gradient
tracking algorithm that, given a smooth strongly convex time-varying cost
function, is capable of tracking its solution sequence
$\{\xstart\}_{t\in\natural}$ in a distributed way up to a bounded error, in line
with time-varying optimization results~\cite{dall2019optimization, Simonetto2020}.
Notice that, this block is a contribution per se to the distributed online
optimization literature. 

In the proposed personalized gradient tracking strategy, the dynamic gradient
tracking update is interlaced with a learning mechanism to let each node learn
the user's cost function $U_i(x)$, by employing noisy user's feedback in the
form of a scalar quantity given by $y_{i,t} = U(x_{i,t}) + \epsilon_{i,t}$,
where $x_{i,t}$ is the local, tentative solution at time $t$ and
$\epsilon_{i,t}$ is a noise term.
It is worth pointing out that in this paper, we consider convex parametric
models, instead of more generic non-parametric models, such as Gaussian
Processes~\cite{Rasmussen2006,
  srinivas2012information, simonetto_personalized_2019,
  Ospina2020}, or convex regression~\cite{Mazumder2019, SimonettoCR2020}. The
reasons for this choice stem from the fact that \emph{(i)} user's functions are
or can be often approximated as convex (see, e.g.,\cite{Johnson2018,Chen2020}
and references therein), which makes the overall optimization problem much
easier to be solved; \emph{(ii)} convex parametric models have better
\emph{asymptotical rate} bounds\footnote{%
  By asymptotical rate, we mean how the approximation gets closer to the true
  function as the number of data points (feedback) increases.
  Shape-constrained Gaussian processes can be used to impose convexity constraints in a practical sense, but their computational complexity scales as $O(t^3)$, where $t$ is the number of data points, they are not trivially extended for decision spaces with dimensions $n>1$, and asymptotical rate bounds are not yet available.  Convex regression has asymptotical rate bounds of the form of $O(t^{-1/n})$, which is very slow compared to the parametric models, and their computational complexity scales at least as $O(t^2 n^3)$.
\label{ft:cc}
  } %
than convex non-parametric models~\cite{Mazumder2019}, which is fundamental when
attempting at learning with scarce data; and \emph{(iii)} a solid online theory
already exists in the form of recursive {least squares
  (RLS)~\cite{ljung1999system, GVK313736715,Sahu2016}}. Therefore, our learning
mechanism is based on a RLS algorithm, whose asymptotical rate is characterized.

Although the high-level algorithmic idea of combining a distributed (online)
optimization update with a recursive regression scheme appears intuitively
reasonable, the concurrent application of the two updates at the same time scale
introduces \emph{several challenges} in the analysis that have been addressed by
properly applying and adapting tools from online and distributed optimization
and from parameter estimation.

To summarize, the main goal of the paper is to provide a first-of-its-kind algorithm to
simultaneously learn and solve optimization problems with unknown convex parametric models online and in 
a distributed fashion, while at the same time incorporating human preferences in the loop. 

\paragraph*{Literature survey}
A centralized bandit framework with a similar
structure to the one considered in this paper has been introduced
in~\cite{simonetto_personalized_2019}, even though in the context of
non-parametric learning (see also references therein for a comprehensive
literature survey). %

In the distributed setup addressed by this paper, we assume that the function
$U_i$ can be modeled as a linearly parametrized convex quadratic function, whose
parameters are unknown and have to be learned. This represents a first step
towards generic parametric models\footnote{The approach in this paper can be
  extended to linearly parametrized convex {functions, but we assume a quadratic
  structure for the sake of clarity. If the user's parametric model is more
complex, we can always focus on local results, where the model is approximately
convex and linear in the parameters}, see also~\cite{Keshavarz2011} for examples of linearly parametrized models applied to inverse control and optimization, which are close in spirit to our problem. \label{foot}}. Non-parametric approaches in the literature to learn unknown functions are e.g., (shape-constrained) Gaussian processes~\cite{Rasmussen2006, Ospina2020} and convex regression~\cite{Mazumder2019, SimonettoCR2020}. As said, we prefer here parametric models for their faster asymptotical rates, cheap online computational load, and ease of introducing convexity constraints.

Another line of research, not followed in this paper, is zero-order (stochastic) online convex optimization, where the cost function is assumed convex, but not known, and its gradient is estimated by function evaluations~\cite{Flaxman2005,Duchi2015}. Even though this line of research is extremely relevant for human-in-the-loop settings (see, e.g.,~\cite{Luo2020}), we distinguish ourselves from it since we do not assume that the user's feedback is available at each time $t$. This is key in human systems where feedback may come intermittently, and still one needs to be able to solve the optimization problem. Imagine for example that a particular user is content with whichever decision and {she/he} does not feel the need for giving feedback, after a few initial ones. Then our algorithm would work seamlessly, since it builds a model for $U_i$, while zero-order methods would still need function evaluations (i.e., feedback) to proceed.   

Regarding optimization problems with (known) time-varying cost function, they have been addressed in 
the distributed optimization literature, both in the stochastic (see, e.g.,~\cite{farina2019randomized,pu2020distributed} 
and references therein) and online/time-varying settings, e.g.,~\cite{Rahili2015,shahrampour2017distributed,
Maros2019,dall2019optimization,Ling2013,bedi_asynchronous_2019,yuan2020can}, and references therein.
Our algorithm relies on the so-called gradient
tracking algorithm firstly proposed in~\cite{varagnolo2016newton,dilorenzo2016next,
  nedic2017achieving,qu2018harnessing}).
The gradient tracking scheme has been originally designed for static 
optimization problems while it has been applied later to 
online problems in, e.g.,~\cite{zhang2019distributed,yuan2020can}. 
The most important difference here is that not knowing either the cost function, the minimum dynamics, 
or both, poses \emph{important additional challenges in ensuring convergence concurrently with learning}.

\smallskip
\paragraph*{Notation}
The $j$-th component of a vector $v$ is $[v]_j$ 
while the $j$-th row of a matrix $A$ is $[A]_j$.
For $m$ vectors $v_1,\dots,v_m$, 
we define $\col(v_1,\dots,v_m) \triangleq [v_1^\top,\dots, v_m^\top]^\top$. 
Given $c\in\R$, $b\in\R^n$ and $A\in\R^{n\times n}$, let 
$v \triangleq \col(c,b,[A]_1^\top,\dots,[A]_n^\top)\in\R^{1+n+n^2}$, then we define 
the operator $\unpack(v)$ so that $(A,b,c) = \unpack(v)$.
The all-one vectors of appropriate dimension is $\one$. 
Gradients w.r.t.~the variable $x$ of the function $f(x;t)$ are indicated with $\nabla f(x;t)$.

\section{Problem Assumptions}\label{sec:set-up}

Problem~\eqref{pb:problem} is to be solved in a distributed way by a network of $N$ agents. We have depicted the problem setting in Figure~\ref{fig:setup}: each agent is composed by a physical node (e.g., a home, a car, a mobile phone) linked to an end-user. The nodes are equipped with a time-varying cost $V_i$ and can evaluate a noisy version of $U_i$ by asking the user for feedback on a particular decision $x_{i,t}$. Each node can compute and communicate with its direct neighbors over a fixed network. In this context, each agent $i$ has only a \emph{partial} knowledge of the target problem. %

\begin{remark}
We assume that the users give feedback at each time $t$ that they are asked for it, with no delay. This is not a limitation: we could consider cases in which users give intermittent feedback at different time-scales and with delays. This would mean that the learning would be slower. From the optimization perspective, since the knowledge of $U_i$ changes every time a new feedback is received, the worst case scenario is when feedback is given at each time (see also~\cite{simonetto_personalized_2019}).\oprocend
\end{remark}

\begin{figure}
\centering
\hspace*{5mm}\includegraphics[width=0.90\columnwidth]{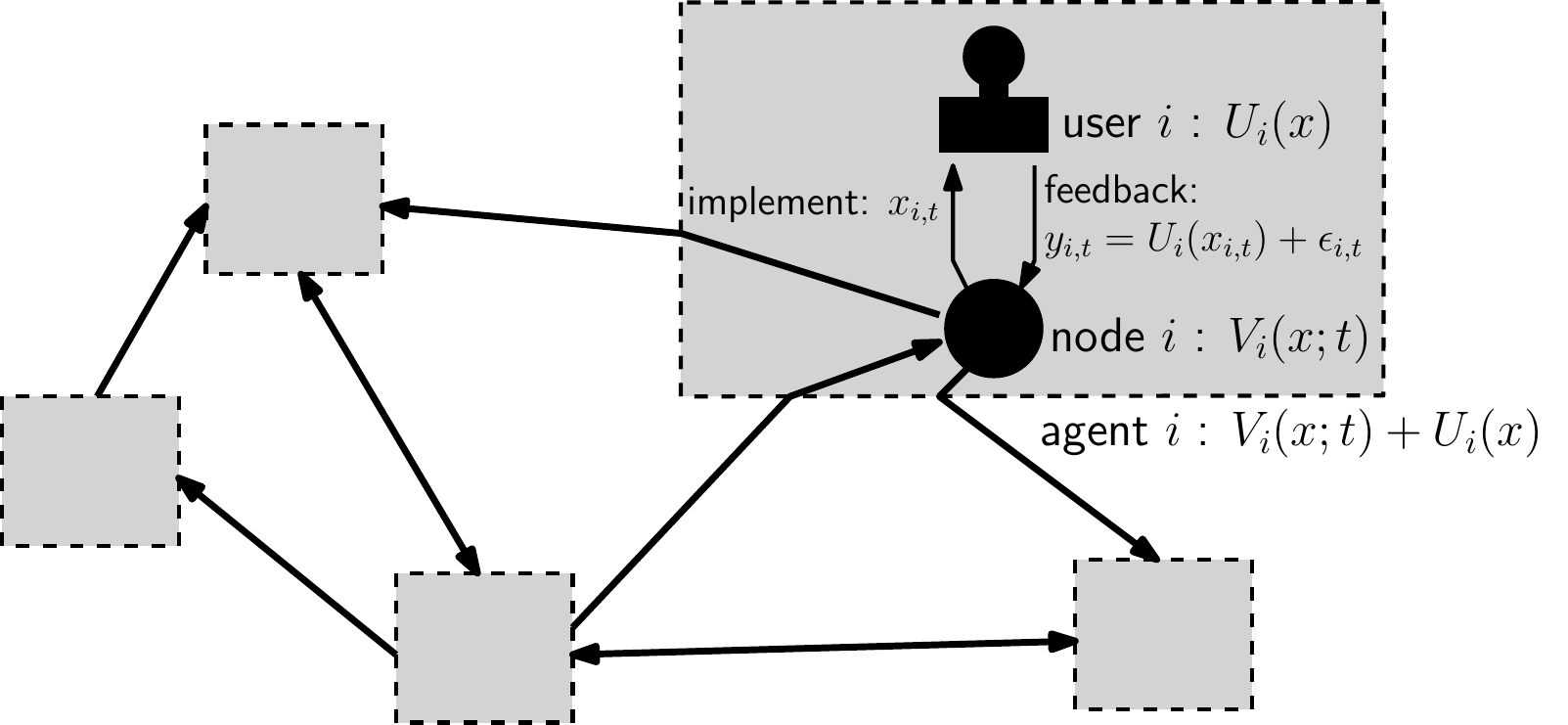}
\caption{The problem setup: a network of connected and communicating nodes, each node with associated an end-user from which feedback may be asked on their own dissatisfaction on a particular decision. }
\label{fig:setup}
\end{figure}

We consider the following assumption on the problem structure.
\begin{assumption}\label{assumption:functions}
  For all $i=1\dots,N$ it holds that: 
  
  (i) The function $V_i(x;t)$ is $m_V$-strongly convex and 
  its gradients are $L_V$-Lipschitz continuous for all $t \in\natural $.

  (ii) The function $U_i(x)$ has a quadratic structure, i.e.,
  $
    U_i(x) = \tfrac{1}{2}\,x^\top P_i x + q_i^\top x + r_i,
  $
  with $P_i \in\R^{n\times n}$ symmetric and with eigenvalues in the range $[m_i, L_i]$, with $L_i \ge m_i >0$, $q_i \in\R^n$, $r_i \in\R$.
  
  (iii) The parameters $P_i$, $q_i$ and $r_i$ of $U_i$ are \emph{unknown}, however one knows a (loose) bound on $L_i$, and noisy measurements %
  of $U_i(x)$ can be taken for any point $x\in\R^n$ as
  $   y_i = U_i(x) + \epsilon_i,
  $%
  where $\epsilon_i$ denotes a generic scalar zero-mean noise with finite variance.
  
  (iv) The optimizer of problem~\eqref{pb:problem}, $\xstart$, is finite for 
  each $t  \in \mathbb{N}$, and $\|\xstart \| < \infty$. 
  \oprocend
\end{assumption}

Assumption~\ref{assumption:functions} on the engineering function $V_i(x;t)$ 
is quite standard in the time-varying literature~\cite{dall2019optimization, Simonetto2020}. %

As for the the $m_i$-strongly convex, $L_i$-smooth quadratic model of $U_i(x)$,
we point out that, though partially restrictive, this structure is reasonable as discussed in 
the introduction (see also Footnote~\ref{foot}) and it can be relaxed. Loose bounds on $L_i$ can be obtained from experiments and average user data. 
Finally, the finiteness assumption on the optimizer (which exists and it is unique for \emph{(i)-(ii)}) just ensures that the problem is well-posed even in a time-varying setting.%

Since each $U_i(x)$ is quadratic but unknown, its parameters {need to} be estimated over time. 
Therefore, we let each agent $i$ consider an approximation of $U_i(x)$ at each time $t$ given by
\begin{align} \label{eq:hatUit}
  \hUit(x) \triangleq 
  \tfrac{1}{2}\, x^\top \hPit x + \hqit^\top x + \hrit,
\end{align}
where $\hPit, \hqit$ and $\hrit$ represent the current estimates of the true (unknown) 
parameters $P_i$, $q_i$ and $r_i$.
We then define the local estimated cost of agent $i$ as
\begin{align*}
  \hf_i(x;t)\triangleq V_i(x;t) + \hUit(x).
\end{align*}

Defining $\hf(x;t) \triangleq \sum_{i=1}^N\hf_i(x; t)$, we denote by $\hfstart$ 
its minimum value and by $\hxstart$ the minimizer. 
Consistently, we define $f(x;t)\triangleq\sum_{i=1}^N f_i(x; t)$
and its minimum value $\fstart$ attained at some $\xstart$.

At this point, we make no specific choice on the type of estimation/learning algorithm 
to determine $\hUit(x)$, provided that it satisfies the following.

\begin{assumption}\label{assumption:bounded_variations}
For the chosen estimation algorithm, the estimated $\hUit(x)$ is bounded for 
any finite $x$, for all $i$ and $t$. Moreover:
 
(i) \emph{With high probability}, the estimated $\hPit$ is symmetric and it has eigenvalues in the set $[0, \mu L_i]$, $\mu>1$. %
I.e., for any $\delta \in(0,1]$ and $\mu>1$, there exists a finite $\bar{t}$, for which:
  $$
\mathbf{Pr}(\mu L_i I_n \ge \hPit \ge 0\,|\, \forall t\geq \bar{t}) \ge 1-\delta,
  $$

(ii) When the first fact holds true, there exist constants $\cUB,\bUB <\infty$ such that:
  \begin{align*}
    & \|\hxstart-\hxstartm \| \le \cUB,
    \\ 
    & \max_i \, \| \nabla \hf_i(\hxstart;t) - \nabla \hf_i(\hxstart;t-1) \|  \le\bUB.
    \eqoprocend
  \end{align*}
\end{assumption}

Assumption~\ref{assumption:bounded_variations}(i) is a mild assumption, and it will hold for our RLS scheme~[Cf. Appendix~\ref{ap:RLS:prop}]. 
It imposes that eventually (and with high probability), the estimated values of $\hPit$ get close to obtain the properties of the true $P_i$.%

Once $\mu L_i I_n \ge \hPit \ge 0$, then the approximate problems are convex 
and for Assumption~\ref{assumption:functions}, the optimizer of $\hat{f}(x;t)$ is finite.
Then, Assumption~\ref{assumption:bounded_variations}(ii) is mild and standard in time-varying optimization: it ensures that the problem changes are bounded. 
This in turn guarantees that one is able to track its solution up to a meaningful error bound.

\begin{remark}
A key aspect in time-varying optimization is the $O(T)$ path length, defined as $P_T = \sum_{t=1}^T \|\hxstart-\hxstartm \|$. This is different from (bandit) online convex optimization which often assumes $P_T = o(T)$ or finite $P_T$. For a $O(T)$ path length, one cannot expect less than finite asymptotic error bounds and $O(T)$ cumulative dynamic regret bounds~\cite{Mokhtari2016,shahrampour2017distributed,dall2019optimization}.  \oprocend
\end{remark}

With Assumption~\ref{assumption:bounded_variations} in place, after $\bar{t}$ 
and for all $t\ge\bar{t}$, the approximate cost function $\hat{f}(x;t)$ is 
$m$-strongly convex and $L$-smooth with
$  m = N m_V,\,%
  L = N L_V + \mu \sum_{i=1}^N L_i,$
with probability $1-\delta$, and the local cost function $\hat{f}_i(x;t)$ is $(L_V+\mu L_i)$-smooth. In addition, and with Assumption~\ref{assumption:functions}(ii), for the gradient $\nabla \hat{f}_i(x;t)$ one has that
\begin{align}
  \| \nabla \hat{f}_i(x;t) - \nabla \hat{f}_i(x;t-1) \| 
  &= 
  \| \nabla \hat{f}_i(x;t) - \nabla \hat{f}_i(x;t-1) \nonumber 
  \\ 
  & \hspace*{-1.5cm} \pm (\nabla \hat{f}_i(\hxstart;t) - \nabla \hat{f}_i(\hxstart;t-1)) \| \nonumber 
  \\
   &\hspace*{-2cm}\leq 2 (L_V+\mu L_i) \| x - \hxstart\| + \bUB,\label{eq:mu}
\end{align}
with probability $1-\delta$. 
In addition, the estimation error $|\hUit(x) - U_i(x)|$ is bounded 
for any finite $x$ since $\hUit(x)$ is proper, 
and one can define the \emph{estimation error length} as,
\begin{equation}
c_U := \sum_{t=1}^T |\hUit(x) - U_i(x)|.
\end{equation}
Under the reasonable assumption that the estimator delivers a bounded error estimation, i.e., $|\hUit(x) - U_i(x)|< c_u < \infty$ for all $i, x, t\geq \bar{t}$, then $c_U = O(T)$. More sensible estimation algorithms will yield $c_U = o(T)$, as we will show.

Regarding the structure of the communication network, 
it is modeled through a weighted graph $\GG=(\VV,\EE, \WW)$ in which $\VV=\until{N}$ 
denotes the set of nodes, $\EE\subseteq\VV\times\VV$ the set of edges 
and $\WW=[w_{ij}]\in\R^{N\times N}$ the weighted adjacency matrix.
We let $\GG$ satisfy the following.
\begin{assumption}\label{assumption:communication}
  The graph $\GG$ is directed and strongly connected.
  The weighted adjacency matrix $\WW$ is doubly-stochastic, i.e., $\sum_{j=1}^N w_{ij}=1$ 
  for all $i= 1,\ldots, N$ and $\sum_{i=1}^N w_{ij}=1$ for all $j= 1,\ldots, N$.
  Moreover, for all $i= 1,\ldots, N$, $w_{ij}>0$ if and 
  only if $j\in\NNii$, where $\NNii\triangleq\{j\mid (j,i)\in\EE\}\cup\{i\}$ is the set of 
  in-neighbors of node $i$.\oprocend
\end{assumption}
The condition above does not include all possible communication topologies, 
however it includes the broad class of balanced digraphs. See~\cite{gharesifard2012distributed}
for further details.

\section{Personalized Gradient Tracking\\ Distributed Algorithm} \label{sec:algorithm}

We describe now our novel distributed online algorithm for solving 
Problem~\eqref{pb:problem}, along with its theoretical properties. %

\subsection{Distributed Algorithm Description}\label{sec:algorithm_description}
Each agent $i$ stores and updates several states. First, it has a local estimate $\xit \in \R^n$ of the 
solution of problem~\eqref{pb:problem} at iteration $t$. Second, it maintains
local estimates $\hPit\in \R^{n\times n}$, $\hqit \in \R^n$ and $\hrit\in \R$ of the 
unknown parameters of the local function $U_i(x)$ {(cf.~\eqref{eq:hatUit})}.
Third, it uses an auxiliary state $\dit\in \R^n$ to
reconstruct an the current value of the gradient of $\sum_{i=1}^N\hf_i(\xit; t)$.

For computational convenience, the local variable $\xit$ will be often arranged in the 
following vectorized form
\begin{align} \notag
  \tvit=\col(1, \xit, [x_{i,t}]_{1}\xit/2, \dots, [x_{i,t}]_{n}\xit/2) \in \R^{1+n+n^2}.
\end{align}
Each iteration $t \in \natural$ of the distributed algorithm consists in three consecutive actions performed 
by each agent $i$.
\begin{enumerate}
  \item A \emph{feedback} on the current local solution estimate $\xit$ is obtained from the user. 
  In particular, a noisy \emph{measurement} of the output of $U_i(\cdot)$ evaluated at $\xit$ is computed 
  and stored as $\yit$ given in~\eqref{eq:feedback}.
  
  \item The estimates $\hPit$, $\hqit$ and $\hrit$ of the unknown 
  parameters $P_i$, $q_i$ and $r_i$ of $U_i$ are updated by means 
  of an ad-hoc \emph{learning} procedure~\eqref{eq:algo_learning}. This procedure relies on 
  a RLS scheme which makes use only of the most updated 
  data $(\yit, \xit)$, thus not requiring to store and use all the past 
  points generated by the distributed algorithm.
  
  \item The local solution estimate $\xit$ of problem~\eqref{pb:problem} at time $t$ is 
  updated via a \emph{dynamic gradient tracking} distributed 
  algorithm~\eqref{eq:gradient_tracking_table}, whose aim is to track the sequence 
  of solutions $\{\xstart\}_{t \in \natural}$ of problem~\eqref{pb:problem}. 
\end{enumerate}

Algorithm~\ref{algo} reports the pseudocode of the proposed scheme,
with step-size $\alpha>0$ and tuning parameter $\eta\gg 0$.
\begin{algorithm}[h!]
  \begin{algorithmic}
    \init $x_{i,0}$ arbitrary, 
    $d_{i,0}=\nabla \hf_i(x_{i,0}; 0)$, 
    $R_{i,0}=\eta I_{1+n+n^2}$,
    $\ha_{i,0} = 0$. %

    \vspace{1ex}
    
    \evol{$t=1,2,\dots$} 

    \State \textsc{Measuring/Feedback}
    \begin{align}\label{eq:feedback}
      \yit &= U_i(\xit) + \eit 
    \end{align}

    \State \textsc{Learning} 
    \begin{subequations}\label{eq:algo_learning}
    \begin{align} 
          \Sit &= \frac{\Ritm \tvit}{ 1 + \tvit^\top \Ritm \tvit } \\
          \Rit &=  \Ritm- (1 + \tvit^\top \Ritm \tvit ) \Sit \Sit^\top \\
          \hait &= \haitm + (\yit - \tvit^\top \haitm) \Sit
    \end{align}
    \end{subequations}
      
      \begin{align}\label{eq:unpack}
        (\hPit, \hqit,\hrit) = \unpack(\hait), \quad \hPit \leftarrow (\hPit + \hPit^\top)/2
      \end{align}

      \vspace{1ex}

    \State \textsc{Dynamic Gradient Tracking}
    \begin{subequations}\label{eq:gradient_tracking_table}
    \begin{align}
    \label{eq:x_update}
      \xit & = \sum_{j \in \mathcal{N}_i} w_{ij}\xjtm - \alpha \, \ditm
      \\
      \label{eq:g_update}
      \git & = \nabla V_i(\xit; t) + \hPit^\top \xit + \hqit
      \\[0.8em]
      \label{eq:d_update}
      \dit & = \sum_{j \in \mathcal{N}_i} w_{ij}\djtm + (\git - \gitm)
    \end{align}
    \end{subequations}
  \end{algorithmic}
  \caption{Personalized Gradient Tracking}
  \label{algo}
\end{algorithm}

\subsection{Parameters Estimation via Recursive Least Squares (RLS)}
\label{sec:RLS}

The aim of the learning part of Algorithm~\ref{algo} (cf.~\eqref{eq:algo_learning}) is to provide
a recursive scheme to let each agent $i$ estimate the unknown parameters of $U_i$. 
Specifically, the considered scheme aims at solving, for each $t$, 
the least squares (LS) problem
\begin{align}\label{pb:estimation}
  \m_{P \in \R^{n\times n}, \, q\in \R^n, \, r \in \R} 
  \: \: 
    \sum_{s=1}^{t}\Big( \tfrac{1}{2}\,\xis^\top P \xis + q^\top \xis + r - \yis\Big)^2,
\end{align}
for a given set of estimate-measurement pairs $(\xis, \yis)_{s=1}^t$.

By defining $\tai \triangleq \col(r, q, [P]_1^\top, \dots, [P]_n^\top)\in\R^{1+ n+n^2}$,
problem~\eqref{pb:estimation} can be equivalently recast into 
\begin{align}\label{pb:estimation_eq}
  \hait = \arg\min_{ \tai } \: \: \sum_{s=1}^{t}(\tai^\top \tvis - \yis)^2,
\end{align}
and $\hPit, \hqit$ and $\hrit$ can be then retrieved from $\hait$ 
via~\eqref{eq:unpack} (cf.~ the Notation) and then made symmetric. 
Now, instead of keeping track of all the data, problem ~\eqref{pb:estimation_eq} is 
solved as data become available by means of a RLS
approach~\cite[Chap.~11]{ljung1999system}, 
yielding~\eqref{eq:algo_learning} in Algorithm~\ref{algo}. 

The estimate computed by using RLS differs from the standard, non recursive, 
least squares (LS) counterpart only in the initial iterations, due to the initialization, 
which is quickly negligible~\cite[Chap.~11]{ljung1999system}; the asymptotic 
convergence properties coincide with those of the non recursive LS approach. 
Upon defining $\xi_{i,\star} = \col(r_i, q_i, [P_i]_1^\top, \dots, [P_i]_n^\top)$ for all $i$, then
for each agent the following classical result holds.

\begin{lemma}[Large sample aymptotic properties of LS]\label{lemma:LS}
  Let the data sequence $\{(\tvis,y_{i,s})\}_{s\ge 0}$ be such that:
  \begin{itemize}
    \item the $\{(\tvis,y_{i,s})\}_{s\ge0}$ is a realization of a jointly stationary and ergodic stochastic process;
    \item the matrix $\Sigma_{xx}=\E[\tvis\tvis^\top]$ is nonsingular;
    \item for $\omega_{i,s} \triangleq \tvis\epsilon_{i,s}$, then $\{\omega_{i,s}\}$ is a martingale 
    difference sequence with finite second moments (cfr.~\cite[Assumption~2.5]{GVK313736715}), 
    and denote $S=\E[\omega_{i,s} \omega_{i,s}^\top]$.
  \end{itemize} 
  Then,
  \begin{align}\label{eq:convergence_distribution}
    \sqrt{t}(\hait - \xi_{i,\star} ) \xrightarrow{D} \NN(0,\Sigma_{xx}^{-1}S\Sigma_{xx}^{-1} ), \text{  as } t\to\infty,
  \end{align}
  where the notation $\xrightarrow{D}$ stands for convergence in distribution. 
\end{lemma}
\begin{proof}
  See, e.g.,~\cite[Prop.~2.1]{GVK313736715} and~\cite[Chap.~8, 9, 11]{ljung1999system}.
\end{proof}

Result~\eqref{eq:convergence_distribution} implies that the random variable 
$\sqrt{t}(\hait-\xi_{i,\star})$ is asymptotically normal distributed, 
and that $\|\hait-\xi_{i,\star}\| \to 0$ with rate $O(1/\sqrt{t})$. 
That is, the rate $O(1/\sqrt{t})$ is the asymptotical rate bound for (R)LS, 
and this will help us show that the estimation length $c_U = O(\sqrt{T})$.

The assumptions in Lemma~\ref{lemma:LS} require some words when applied to our setting.
Since the regressors $\tvis$ are determined by the gradient tracking process, and 
ultimately (upon convergence) they are close to the optimizer trajectory, we are requiring 
that the optimizers $\{\x_{\star}(t)\}$: (i) eventually behave as a stationary and 
ergodic process, and (ii) are never exactly the same (so that $\Sigma_{xx}$ remains non-singular).
In practice in our model the optimizers change in time due to external, time-varying data-streams 
(which could be assumed stationary and ergodic) and, thus, satisfy this assumption.

\subsection{Dynamic Gradient Tracking} \label{sec:dyn_GT}

The step in~\eqref{eq:gradient_tracking_table} is meant to implement
a gradient tracking distributed algorithm tailored for an online optimization 
problem, whose convergence is provided next.
\begin{theorem}\label{thm:GT_convergence}
  Consider the sequence $\{\xit\}_{t\ge 1}$ generated by~\eqref{eq:gradient_tracking_table} 
  and let $\avxt\triangleq \frac{1}{N}\sum_{i=1}^N \xit$. 
  Let Assumptions~\ref{assumption:functions},~\ref{assumption:bounded_variations}, 
  and~\ref{assumption:communication} hold. Choose a $\mu > 1$. Then, there exist a $\rho<1$ and a small enough step-size $\alpha$ in $(0, N/L]$, for which the following holds with high probability
  \begin{align*}
    \limsup_{t \to \infty} 
    \sum_{i=1}^N \hf_i(\avxt;t) - \hfstart 
    = 
    \frac{L (N \bUB^2 + \cUB^2)}{2 (1-\rho)^2} =: \frac{L}{2}\bar{c}^2 %
  \end{align*}
  with linear rate $\rho$. The consensus metric $C_{T}$
  satisfies
   $\limsup_{T \to \infty} 
    C_{T}(\{x_{i,T}\}_{i=1}^N, \bar{x}_{T}) 
    = 
    \bar{c}^2$, and the average $\bar{x}_{T}$ is bounded.

  ~\oprocend
\end{theorem}

The proof of Theorem~\ref{thm:GT_convergence} is given in Appendix~\ref{app:dynamic_GT}. %
The result is in line with current works in time-varying optimization~\cite{dall2019optimization, Simonetto2020}, as well as regret results with dynamic comparators when the path length grows as $O(T)$ and we employ a constant step-size~\cite{shahrampour2017distributed}.

\subsection{Regret Analysis of Algorithm~\ref{algo}} \label{sec:regret}

The next theorem, whose proof is reported in Appendix~\ref{app:regret}, represents the second main 
result of this paper. 
It shows that a bound on the cumulative regret can be provided under suitable assumptions, 
and that the asymptotic average regret is bounded.
\begin{theorem}\label{thm:regret}
  Let the sequences $\{(\tvit, \yit)\}_t$ be generated by Algorithm~\ref{algo}. 
  Let Assumptions~\ref{assumption:functions},~\ref{assumption:bounded_variations}(ii) 
  and~\ref{assumption:communication} hold. 
  Choose a $\mu > 1$. Then, there exist a $\rho<1$ and a small enough step-size $\alpha$ in $(0, N/L]$, 
  for which w.h.p.
  \begin{multline*}
   R_T(\{\avxt\}_{t=1}^T)
    \le 
    O(1) + O (c_U) + O\left( T\, \frac{L (N \bUB^2 + \cUB^2)}{2 (1-\rho)^2}\right).
  \end{multline*}
  Moreover, w.h.p., the average dynamic regret reaches an asymptotical value as
  \begin{align*}
  \limsup_{T\to\infty} \, \frac{R_T(\{\avxt\}_{t=1}^T )}{T} = O\left(\frac{c_U}{T}\right) + \frac{L}{2}\bar{c}^2 = O(1).%
  \end{align*} 
  Finally, w.h.p., the consensus metric $C_{T}$ is such that $\limsup_{T \to \infty} 
    C_{T}(\{x_{i,T}\}_{i=1}^N, \bar{x}_{T}) 
    = 
    \bar{c}^2$.
\oprocend
\end{theorem}
Algorithm~\ref{algo} delivers a bounded average 
dynamic regret with high probability. In particular, the dynamic regret is composed  of three terms.
The first $O(1)$ term collects the initialization errors (e.g., when $\hUit$ is nonconvex).
The second $O(c_U)$ term, more standard, represents the learning bound. 
(It is in general $O(T)$, but $O(\sqrt{T})$ if the assumptions of Lemma~\ref{lemma:LS} are 
verified, see Appendix~\ref{ap:RLS:prop} (Lemma~\ref{lemma:RLS}), thereby 
vanishing as $O(1/\sqrt{T})$ in the average regret result).
Finally, the third $O(T)$ term pertains the tracking of the distributed solution trajectory, and 
it is linear in $T$ since the path length is linear in $T$~\cite{shahrampour2017distributed}. 
The asymptotical bound depends on how fast the problems are changing in time, 
due to variations of the gradients and the optimizers, as typical in time-varying optimization. 
Finally, note that Assumption~\ref{assumption:bounded_variations}(i) is not required here, since it is verified for our RLS scheme~[Cf. Appendix~\ref{ap:RLS:prop}].

\begin{remark}[Regret in a distributed setting]
  Under boundedness of the consensus metric $C_T$
  given by Theorem~\ref{thm:regret},
  an agent $j$-specific regret bound $\sum_{t=1}^T \sum_{i=1}^N f_i (x_{j,t}; t)  - \fstart $ 
  can also be derived, with the same convergence rate, and leading term 
  of $O(c_U) + O(2 T \bar{c}^2)$ 
  (Cf.~\cite[Appendix~\ref{extra-bound}]{notarnicola2020personalized_arxiv}).\oprocend  
\end{remark}

\subsection{Computational and communication complexity}\label{sec:comm}

We finish our analysis of Algorithm~\ref{algo} by reporting its computational and communication complexity. First, only local computations are carried out, and the most demanding are matrix/vector multiplications on vector $\tvit \in \real^{1+n+n^2}$, delivering a computational complexity of $O(n^4)$. This is in comparison with Gaussian Processes $O(t^3)$ and convex regression $O(t^2 n^3)$ [Cf. Footnote~\ref{ft:cc}]. This makes our method less computational intensive than other techniques, especially for large $t \gg n$ (i.e., as more and more data comes in). This is due to the fact that our method is recursive. 

As for the communication complexity, our gradient tracking employs two communication rounds for each iteration for a total of at worst $4 (N-1) n$ scalar sent.

\section{Numerical Example}\label{sec:example}

We consider a scenario with  both $V_i$ and $U_i$ quadratic, i.e.,
\begin{align*} %
    \m_{x\in\R^3}
    \: &\: \sum_{i=1}^N \Big(\underbrace{\|x - p_i(t)\|^2}_{V_i(x;t)} + \underbrace{\|x- v_i\|^2}_{U_i(x)}\Big), \qquad t \geq 0.
\end{align*}
with $p_i(t) \in \R^3$ for all $t$ and $v_i\in\R^3$.

We implemented the Personalized Gradient Tracking Algorithm~\ref{algo} with DISROPT~\cite{farina2019disropt} 
and performed a simulation with $N=30$ agents, in which each target speed $p_i(t)$ evolves according to the following law 
\begin{align*}
  p_i(t) = z_i + \psi_i \sin(t/m_i)
\end{align*}
with $z_i\in\R^3$, $\psi_i \in \R^3$ and $m_i>1$. 
We randomly generate the coefficients by picking
$v_i \in \mathcal{U}[-1.5, 1.5]^3$, 
$z_i \in \mathcal{U}[-5, 5]^3$, 
$m_i\in \mathcal{U}[100,150] \cap \natural$, 
$\psi_i\in\mathcal{U}[0.5, 0.6]$ 
and $\epsilon_{i,t}\in \mathcal{N}(0, 0.2)$ for all $i=1,\dots, N$. 
We ran $10^6$ iterations with step-size $\alpha = 0.01$ 
and initial conditions $x_{i,0}\in\mathcal{U}[-1.5, 1.5]^2$. %
The evolution of the average regret $R_t/t$ obtained by Algorithm~\ref{algo} is 
shown in Figure~\ref{fig:regret}. Specifically, we evaluate the dynamic regret 
as expressed in~\eqref{eq:regret} at $\avxt\triangleq \frac{1}{N}\sum_{i=1}^N \xit$ 
for all $t = 1,\ldots, 10^6$.
As expected from Theorem~\ref{thm:regret}, the average regret decays to some constant value.
\begin{figure}[!htbp]
\centering
     \includegraphics[width=0.9\columnwidth]{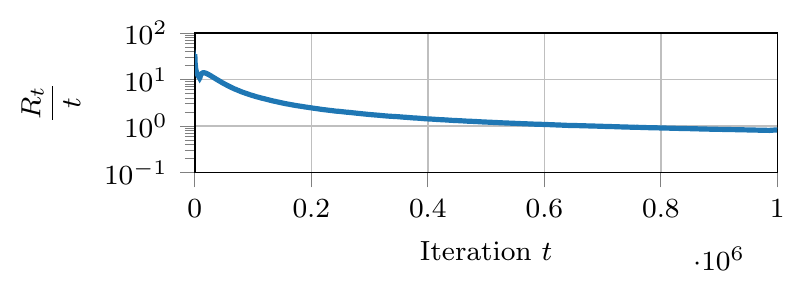}
     \caption{Evolution of the average regret.}
     \label{fig:regret}
\end{figure}

Figure~\ref{fig:consensus_tracking_error} shows the 
consensus and tracking error. In particular, it can be appreciated that they become 
stationary, though not vanishing, after the initial transient highlighted in the insets, consistently 
with the theoretical bound proved by~\eqref{con:lim}.
\begin{figure}[!htbp]
\centering
     \includegraphics[scale=.9]{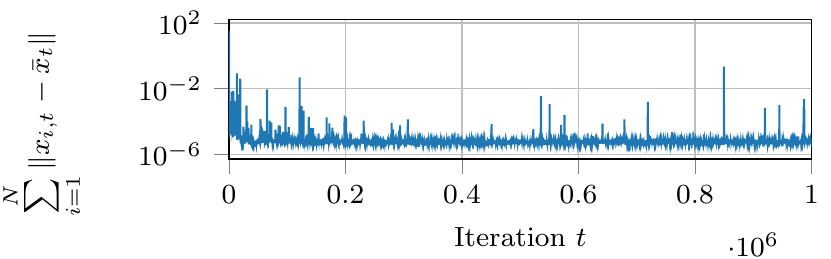}
     \includegraphics[scale=0.9]{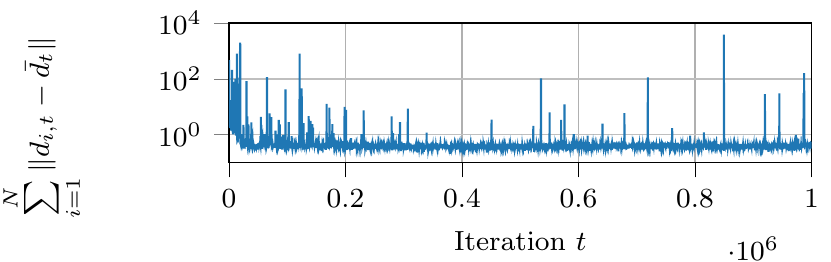} %
     \caption{Evolution of the consensus error (top) and the tracking error (bottom).} %
     \label{fig:consensus_tracking_error}
\end{figure}
\vskip-.5cm
\section{Conclusions}
 In this paper, we addressed the problem of solving in a distributed way an online optimization problems 
 in which the local cost functions are composed by a known and an unknown part. 
 We proposed an algorithm that concurrently tracks the solution of the problem and estimates 
 the parameters of the unknown portion of the objective function. 
 Finally, we showed that a bounded (possibly vanishing) average regret is achieved under 
 suitable assumptions. A numerical example is provided to corroborate the theoretical results.

\bibliographystyle{IEEEtran}
\bibliography{biblio_personalized_online}

\appendix

\section{Appendix}
Let $\avxt \triangleq \frac{1}{N}\sum_{i=1}^N \xit$, 
$\avgt \triangleq  \frac{1}{N}\sum_{i=1}^N \nabla \hf_i(\avxt;t)$
and 
$\avdt \triangleq  \frac{1}{N}\sum_{i=1}^N \dit$ be
the averages of the local quantities in~\eqref{eq:gradient_tracking_table}
for all $t\ge0$. Through simple manipulations, we obtain
\begin{align}\label{eq:xy_avg_update}
    \avxt \!=\! \avxtm \!-\! \alpha\avdtm, \quad 
    \avdt \!=\! \avdtm \!+ \!\frac{1}{N} \sum_{i=1}^N \left( \git\! - \!\gitm\right).%
\end{align}
By exploiting the (column) stochasticity of the weights (cf.~Assumption~\ref{assumption:communication}), 
and the initialization $d_i^0=\nabla \hf_i(x_i^0; 0)$ it can be shown that, for all $t\geq 0$,
\begin{align}\label{eq:conservation_y}
    \avdt = \frac{1}{N} \sum_{i=1}^N \git = \frac{1}{N} \sum_{i=1}^N \nabla \hf_i ( \xit , t). 
\end{align}

Moreover, letting $\bxt \triangleq \col( x_{1,t},\dots,x_{N,t})$, 
$\bdt \triangleq \col (d_{1,t},\dots,d_{N,t} )$ and $\bgt \triangleq \col ( g_{1,t}, \dots,g_{N,t} )$,
algorithm~\eqref{eq:gradient_tracking_table} can be restated as
\begin{align}\label{eq:xy_update_stack}
    \bxt \!=\! W \bxtm \!- \!\alpha \bdtm,\quad
    \bdt\! =\! W \bdtm + ( \bgt\! -\! \bgtm).
    \end{align}
where $W \triangleq \WW \otimes I_n$ with $\otimes$ denoting the Kronecker product.

\subsection{Intermediate  Results}
The analysis relies on properties of
    the consensus error $\|\bxt \!-\! \one \avxt \|$,
  the tracking error $\|\bdt \!-\! \one \avdt \|$
  and   the optimality error $\|\avxt \!-\! \xstart \|$.
This is captured in the next lemmas whose proofs are
provided later.

\begin{lemma}\label{lemma:consensus_error}
    Let assumption~\ref{assumption:communication} hold. Then,
    for all $t \ge 0$,
    \begin{align*}
      \|\bxt - \one \avxt \|
      \leq 
      \sigma_{W}
      \| \bxtm -\one \avxtm \| + \alpha \|\bdtm - \one \avdtm \|
    \end{align*}
    where $\sigma_W$ be the spectral radius of $W-\tfrac{1}{N} \one\one$.\oprocend
\end{lemma}

\begin{lemma}\label{lemma:optimality_error}
  Let Assumptions~\ref{assumption:functions},~\ref{assumption:bounded_variations},~\ref{assumption:communication} hold. Then, for $t\geq \bar{t}$ and with probability $1-\delta$:
  \begin{align*}\!
    \|\avxt \! -\! \hxstart \|\!
    \leq\!
    \theta\|\avxtm \! -\! \hxstartm \|
    \! +\! 
    \alpha\tfrac{L}{\sqrt{N}}\|\bxtm \! -\! \one \avxtm \| 
    \! +\! 
    \cUB
  \end{align*}
  with $\theta=\max\{ |1- L \alpha/N|,|1-m \alpha/N|\}$.\oprocend
\end{lemma}
\begin{lemma}\label{lemma:tracking_error}
  Let Assumptions~\ref{assumption:functions},~\ref{assumption:bounded_variations},
  \ref{assumption:communication} hold. 
  Then, for $t\geq \bar{t}$ and with probability $1-\delta$:
  \begin{align*}
    \|\bdt - \one \avdt \|
    &
    \leq (\sigma_W+\alpha L)\|\bdtm - \one \avdtm \|+\nonumber
    \\
    &\hspace{-8ex}+(L\|W-I\|+ 2L + \alpha L^2\sqrt{N})\|\bxtm- \one \avxtm \| \nonumber
    \\
    &\hspace{-8ex}+ (2L\sqrt{N}+\alpha L^2\sqrt{N}) \|\avxtm- \xstartm \|+\sqrt{N} \bUB.
    \eqoprocend
  \end{align*}
\end{lemma}

\subsection{Proof of Theorem~\ref{thm:GT_convergence}}
\label{app:dynamic_GT}
Let us define
\begin{align*}
    \Vt : =\begin{bmatrix}
        \|\avxt - \hxstart \| \\
        \|\bxt - \one \avxt \| \\
        \|\bdt - \one \avdt \|
    \end{bmatrix},\quad
    \Zt : =\begin{bmatrix}
        \bUB\\
        \cUB
    \end{bmatrix}.
\end{align*}

By combining Lemma~\ref{lemma:consensus_error},~\ref{lemma:optimality_error} and~\ref{lemma:tracking_error}, 
we have that 
\begin{align}\label{eq:recursive}
  \Vt 
  \leq 
  A(\alpha) \Vtm + B \Ztm
\end{align}
for $t\ge \bar{t}$ and with probability $1-\delta$,
where 
\begin{align}\notag
    A(\alpha) : =
    \begin{bmatrix}
        \theta &  \alpha \frac{L}{\sqrt{N}} & 0  \\
        0 & \sigma_W & \alpha \\
        a_1 & a_2 & \sigma_W + \alpha L\\
    \end{bmatrix}, 
    \quad 
    B
    :=
    \begin{bmatrix}
    0 & 1\\
    0 & 0 \\
    \sqrt{N} & 0
  \end{bmatrix},
\end{align}
with
$a_1 = \alpha L^2\sqrt{N}\!+\! 2L\sqrt{N}$
and $a_2 = L\|W\!-\!I\|\!+\! 2L\!+\!\alpha L^2\sqrt{N}$.

Now, since by assumption $\alpha\leq N/L$ and $m\leq L$, we have that $\theta = 1-\alpha m/N$ and hence
\begin{multline}\notag
    A(\alpha)= \begin{bmatrix}
        1 & 0 & 0   \\
        0  & \sigma_W & 0 \\
        2L\sqrt{N} & L\|W-I\|+2L & \sigma_W  \\
    \end{bmatrix} \\ 
    +\alpha 
    \begin{bmatrix}
        -\frac{m}{N} & \frac{L}{\sqrt{N}} & 0   \\
          0 &  0 &  1\\
        L^2\sqrt{N} &  L^2\sqrt{N} & L  \\
    \end{bmatrix}.
\end{multline}

We use now \cite[Theorem~6.3.12]{horn2012matrix} for a small perturbation $\alpha>0$. 
For $\alpha = 0$, the eigenvalues of $A(\alpha)$ are $1$ and $ \sigma_W<1$. 
By continuity of the eigenvalues w.r.t.~the matrix coefficients, for small 
enough $\alpha$, the eigenvalues $<1$ will remain $<1$. 
For the single eigenvalue $1$ with left eigenvector $\col(1,0,0)$ and right eigenvector 
$\col(1,0,2L\sqrt{N}/(1-\sigma_W))$, one can use \cite[Theorem~6.3.12(i)]{horn2012matrix}, 
to say that the corresponding eigenvalue of $A(\alpha)$, say $\lambda(\alpha)$, 
will be $|\lambda(\alpha) - 1 + \alpha m/N | \leq \alpha \epsilon$
for any $\epsilon>0$ and sufficiently small $\alpha$.
If then one selects e.g., $\epsilon = \frac{m}{2N}$, then 
$
\lambda(\alpha) \in [1-3 \alpha \frac{m}{2N}, 1-\alpha \frac{m}{2N}],
$
meaning that there exists a small enough $\alpha$, for which all the eigenvalues of $A(\alpha)$ are 
all strictly less than one, and therefore the spectral radius of $A(\alpha)$, say $\rho$, becomes 
strictly less than one.
Also, we notice that the input $\Zt$ is bounded.
Since $\Vt, A, B, \Zt$ have nonnegative entries, we can expand~\eqref{eq:recursive} from $\bar{t}$ and get
 $ \Vt 
  \leq 
  A(\alpha)^{t-\bar{t}} \Vtb + \sum_{\tau=\bar{t}}^{t-1} A(\alpha)^{t-1-\tau}B \Ztau $.
Given Assumptions~\ref{assumption:functions}-\ref{assumption:bounded_variations}, for any finite $\bar{t}$, $\|\Vtb\|$ is bounded. 
Therefore we can write
\begin{align} \nonumber
  \| \Vt  \| 
  & \leq 
  \| A(\alpha)^{t-\bar{t}} \Vtb \| + \Big  \| \sum_{\tau=\bar{t}}^{t-1} A(\alpha)^{t-1-\tau}B \Ztau \Big \| 
  \\ 
    & \leq 
  \rho^{t-\bar{t}} \| \Vtb \| + \sum_{\tau=\bar{t}}^{t-1} \rho^{t-1-\tau} \sqrt{N \bUB^2  + \cUB^2}. \label{con:lin}
 \end{align}
And, taking the limit superior:
\begin{align}\label{con:lim}
  \limsup_{t \to \infty} \| \Vt  \|  = \frac{1}{1-\rho} \sqrt{N \bUB^2  + \cUB^2} =:\bar{c}.  
\end{align}
Eq.~\eqref{con:lin} shows that the first term decreases linearly with rate $\rho$ equal to the 
spectral radius of $A(\alpha)$, while the second term is bounded. 
Eq.~\eqref{con:lim} completes the argument yielding the upper limit of the sequence. This finishes the first part of the proof.

The second part of the proof is based on similar arguments to those used in~\cite[Theorem~1]{qu2018harnessing}.
In particular we have that all the entries of $\Vt$ converges to $\bar{c}$ linearly with rate $O(\rho^k)$. 
Moreover, by exploiting the Lipschitz continuity of the gradients of $\hat{f}(x;t) = \sum_i \hat{f}_i(x;t)$ one has
$
  \hat{f}(\avxt;t)-\hat{f}( \hxstart ;t) 
  \leq 
  \nabla \hat{f}( \hxstart ;t)^\top(\avxt-\hxstart ) 
  +\frac{L}{2}\|\avxt - \hxstart \|^2$.
Now, since $\nabla \hat{f}( \hxstart ;t)=0$ the above implies that 
$
  \hat{f}(\avxt;t) -\hfstart 
  \leq \frac{L}{2}\|\avxt - \hxstart \|^2 \le \frac{L}{2}\|\Vt \|^2$
and hence the $\limsup$ of $\hat{f}(\avxt;t) - \hfstart $ converges linearly to $\frac{L\bar{c}^2}{2}$.
It yields
\begin{equation*}
    \mathbf{Pr}\left(\limsup_{t \to \infty} 
    \sum_{i=1}^N \hf_i(\avxt;t) - \hfstart 
    = 
    \frac{L (N \bUB^2 + \cUB^2)}{2 (1-\rho)^2}\right) \geq 1-\delta.
\end{equation*}
This concludes the second part of the proof. 

The third part concerns the convergence of the consensus metric $C_T$. 
By definition $C_T \leq \|\Vt \|^2$ so that the thesis follows.

The fourth part concerns the boundedness of $\|\avxt\|$, which is bounded by the discussion above as
\begin{equation*}
\|\avxt\|\leq \|\avxt-\hxstart \| + \|\hxstart\| \leq \|\Vt\|+\|\hxstart\|,
\end{equation*}
and since $\|\Vt\|$ is bounded for discussion above and $\|\hxstart\|$ is finite by assumption, then $\|\avxt\|$ is bounded.

Finally, since the above limit results are valid for any $\delta\in(0,1]$, we have that they hold with high probability.

\subsection{Asymptotical bounds for RLS}\label{ap:RLS:prop}

Next we present useful asymptotical bounds for RLS that are necessary 
for Theorem~\ref{thm:regret}. 
Their proofs are given later.
\begin{lemma}\label{lemma:assumption}
Assumption~\ref{assumption:bounded_variations}(i) holds for our RLS scheme.\oprocend
\end{lemma}
\begin{lemma}\label{lemma:RLS}
For an estimator satisfying the assumptions of Lemma~\ref{lemma:LS}, 
for any bounded vector $x \in\real^{n}$, the functional learning is bounded as
$| \hUit(x) - U_i(x) | \le O(1/\sqrt{t})$.\oprocend
\end{lemma}

\subsection{Proof of Theorem~\ref{thm:regret}}
\label{app:regret}
Recalling the definition of the cumulative dynamic regret in~\eqref{eq:regret}, 
we can write
\begin{align*}
 & 
 R_T(\{\avxt\}_{t=1}^T)\!
 = 
 \!
 \sum_{t=1}^T \! 
 \left ( \!\sum_{i=1}^N \! \left ( \hf_i( \avxt; t) \!+\! U_i(\avxt) \!-\!  \hUit( \avxt) \! \right ) \!-\! \fstart \! \right )\nonumber
 \\
 & \le \sum_{t=1}^T \Big ( \sum_{i=1}^N \hf_i( \avxt; t) - f_\star(t) \Big )
+ \sum_{t=1}^T \sum_{i=1}^N \Big \vert U_i(\avxt) -  \hUit( \avxt) \Big \vert
\end{align*}
Now, fixing a $\delta\in(0,1]$ one determines a $\bar{t}$, and the first term on the right-hand side 
can be split as
$
  \sum_{t=1}^T ( \sum_{i=1}^N \hf_i( \avxt; t) - f_\star(t)  ) 
  =
  \sum_{t=1}^{\bar{t}-1} ( \sum_{i=1}^N \hf_i( \avxt; t) - f_\star(t) )
  +
  \sum_{t=\bar{t}}^T ( \sum_{i=1}^N \hf_i( \avxt; t) - f_\star(t) )
$,
where in the first $\bar{t}$ iterations the functions $\hf_i$, in general, could have been nonconvex, 
while they are convex after $\bar{t}$ with probability $1-\delta$. 
Notice now that by Assumptions~\ref{assumption:functions}-\ref{assumption:bounded_variations},
both $V_i(x;t)$ and $\hUit(x)$ are bounded for all bounded $x$ and all $i$ and $t$. 
Moreover, by Theorem~\ref{thm:GT_convergence}, $\|\avx_t\|$ is uniformly bounded. 
Thus, we can bound the quantity $\sum_{t=1}^{\bar{t}-1} \left ( \sum_{i=1}^N \hf_i( \avxt; t) - f_\star(t) \right )$ by $O(\bar{t})$. Then, 
\begin{align}
\notag
  R_T(\{\avxt\}_{t=1}^T)
  &
  \le 
  O(\bar{t}) 
  + \sum_{t=\bar{t}}^T \Big ( \sum_{i=1}^N \hf_i( \avxt; t) - f_\star(t) \Big )
  \\
  &\qquad 
  + \sum_{t=1}^T \sum_{i=1}^N \Big \vert U_i(\avxt) -  \hUit( \avxt) \Big \vert.
\label{eq:RT}
\end{align}

Now, we can use the fact that
\begin{align*}
f_\star(t) %
& = \hf_\star(t) + \underbrace{(\hf(\xstart; t)-\hf_\star(t))}_{(I)} 
+ \underbrace{(f_\star(t) - \hf(\xstart; t))}_{(II)}.
\end{align*}
In addition, by strong convexity of $\hat{f}$ and optimality, 
it holds $(I) \ge \frac{m}{2} \|\xstart-\hxstart \|^2 \ge 0$,
while 
$(II) \ge - \sum_{i=1}^N \left\vert U_i(\xstart) - \hUit(\xstart)\right \vert$.
Putting these facts together %
in the expression of the dynamic regret~\eqref{eq:RT}, then, 
\begin{align}%
  R_T(\{\avxt\}_{t=1}^T)
  \le 
  O(\bar{t}) +O(c_U) + \sum_{t=\bar{t}}^T \Big ( \sum_{i=1}^N \hf_i( \avxt; t) - \hf_\star(t) \Big ).
\label{eq:RT1}
\end{align}
Now, by using Theorem~\ref{thm:GT_convergence}, the second term can be upper bounded as
\begin{multline}
  \sum_{t=\bar{t}}^T \Big ( \sum_{i=1}^N \hf_i( \avxt; t) - \hf_\star(t) \Big)
  \\[-0.1cm] 
  \le
  \sum_{t=\bar{t}}^T O(\rho^{t-\bar{t}}) + 
  O\Big ( (T-\bar{t})\frac{L (N \bUB^2 + \cUB^2)}{2 (1-\rho)^2} \Big),
  \label{eq:hf_bound}
\end{multline}
Hence, by combining~\eqref{eq:RT1} and~\eqref{eq:hf_bound} we have that, with probability $1-\delta$,
\begin{align*}
  R_T(\{\avxt\}_{t=1}^T)
  & 
  \le O(1) + O (c_U) + O \Big( T\, \frac{L (N \bUB^2 + \cUB^2)}{2 (1-\rho)^2} \Big) 
\end{align*}
where we used the fact that, since $\bar{t}$ is finite, $O(\bar{t}) +  \sum_{t=\bar{t}}^T O(\rho^{t-\bar{t}}) = O(1)$. 
Since the above is valid with probability $1-\delta$, for any $\delta \in (0,1]$, it is valid with high probability.

As for the consensus metric $C_T$, everything goes as in the proof of Theorem~\ref{thm:GT_convergence}, with the difference to be valid with high probability, which concludes the proof.

\subsection{Proof of Lemma~\ref{lemma:consensus_error}}

  By using~\eqref{eq:xy_avg_update} and~\eqref{eq:xy_update_stack}, one has
  \begin{align*}
    \|\bxt- \one \avxt \| & = \| W \bxtm-\alpha\bdtm - \one \avxtm + \alpha \one \avdtm \|.
  \end{align*}
  Now, by using the triangle inequality and exploiting Assumption~\ref{assumption:communication} we get
  \begin{align*}
    \|\bxt - \one \avxt \|
    &\leq \|W\bxtm- \one \avxtm \| + \alpha \|\bdtm - \one \avdtm \|\\
    &\leq \sigma_{W}\|\bxtm- \one \avxtm \| + \alpha \|\bdtm - \one \avdtm\|
  \end{align*}
  thus concluding the proof.

\subsection{Proof of Lemma~\ref{lemma:optimality_error}}

  By using~\eqref{eq:xy_update_stack} one has
  \begin{align*}
    \|\avxt - \hxstart\|
    & =\|\avxtm-\alpha\avdtm-\hxstart\|
    \\
    &=\|\avxtm-\alpha\avdtm \pm \hxstartm -\hxstart\|
    \\
    &
    \leq \|\avxtm-\alpha \avdtm-\hxstartm \| + \cUB
  \end{align*}
  where in the second line we add and subtract $\hxstartm$ and in the last one we 
    exploit Assumption~\ref{assumption:bounded_variations} with the triangle inequality. 
    Now, by adding and subtracting $\alpha \avgtm$ inside the norm, we get
  \begin{align}
    & \|\avxt- \hxstart \|\nonumber
    \\
    &\leq 
    \left \| \avxtm-\alpha\avdtm - \hxstartm - \alpha \avgtm+\alpha \avgtm\right \| + \cUB \nonumber
    \\
    &
    \leq 
    \left \| \avxtm-\alpha \avgtm - \hxstartm \right \| 
      + \left \| \alpha \avgtm-\alpha\avdtm \right \| + \cUB\nonumber
    \\
    &\leq \theta\|\avxtm-\hxstartm \|+ \left \| \alpha \avgtm-\alpha\avdtm \right \| + \cUB \nonumber
    \\
    &\leq \theta\|\avxtm-\hxstartm \|\!+\!\alpha\frac{L}{ \sqrt{N}}\|\bxtm-\one \avxtm \| + \cUB
    \end{align}
    in which in the second line we use the triangle inequality, 
    in the third one we exploit the convergence rate result for a gradient iteration 
    applied to a smooth and strongly convex function 
    and in the last one we use the Lipschitz continuity of each $\nabla \hf_i$ 
    and~\eqref{eq:conservation_y} to write 
    \begin{align*}
      \|\avgt-\avdt \|
      & =
      \left  \|\frac{1}{N} \sum_{i=1}^N \nabla \hf_i(\avxt; t) - \nabla\hf_i(\xit; t) \right \| 
      \\ 
      & \le 
      \frac{1}{N}\sum_{i=1}^N \| \nabla \hf_i(\avxt; t) - \nabla\hf_i(\xit; t)\| 
      \\
      & \le 
      \frac{1}{N}\sum_{i=1}^N (L_{V}+\mu L_i) \| \avxt - \xit \| 
      \\
      & \le 
      \frac{1}{N} \max_i (L_{V}+\mu L_i) \sum_{i=1}^N \| \avxt - \xit \| 
      \\
      & \le 
      \frac{1}{N} \max_i (L_{V}+\mu L_i) \sqrt{N} \| \one \avxt - \bxt \| 
    \end{align*}
    where the last inequality follows from concavity of the square root function.

\subsection{Proof of Lemma~\ref{lemma:tracking_error}}

  From~\eqref{eq:xy_avg_update} and~\eqref{eq:xy_update_stack}, one has
  \begin{align}
    \|\bdt - \one \avdt\|
    & =
    \Bigg\Vert W \bdtm + \left( \bgt - \bgtm\right) \nonumber
    \\
    & \hspace{3ex}
    -
    \one \left(\avdtm + \frac{1}{N} \sum_{i=1}^N \left( \git - \gitm\right)\right)\Bigg\Vert
  \end{align}
  
  Now, by using the triangle inequality, we have
  \begin{align}  
    &\|\bdt-\one \avdt \|\nonumber
    \\
    &\leq\| W \bdtm - \one \avdtm \|
      + \left \Vert \left(I-\tfrac{1}{N}\one\one^\top\right) \left( \bgt - \bgtm\right)\right\Vert\nonumber
    \\
    &\leq\sigma_W\|\bdtm - \one \avdtm \|
    + \left \Vert I-\tfrac{1}{N}\one\one^\top \right\Vert \|\bgt - \bgtm\|\nonumber\\
    &\leq\sigma_W\|\bdtm -\one \avdtm \|+\|\bgt - \bgtm\|
    \label{eq:p1}
  \end{align}
  where, in the second inequality we used the contraction property of the consensus matrix $W$ 
  and the sub-multiplicativity of $2$-norm 
  and in the third inequality the fact that $\Vert I-\frac{1}{N}\one\one^\top\Vert\leq 1$.
  Now define
  \begin{align*}
    \btgt \triangleq \col ( \nabla \hf_1(\xitm; t), \dots, \nabla \hf_N(\xitm; t) ).
    \end{align*}  
  By adding and subtracting $\btgt$ inside the second term of~\eqref{eq:p1} and 
  using the triangle inequality we have
  \begin{align}
    & \|\bdt - \one \avdt \| \nonumber\\
    &\leq\sigma_W \|\bdtm -\one\avdtm \|+\| \bgt - \btgt\|+\|\btgt- \bgtm\|\nonumber
    \\
    &\leq\sigma_W \|\bdtm -\one\avdtm \|+L\| \bxt-\bxtm\|+\|\btgt - \bgtm\|\nonumber
    \\
    &\leq\sigma_W \|\bdtm -\one\avdtm \|+L\| W\bxtm-\alpha\bdtm-\bxtm\| \nonumber
    \\ 
    & \hspace{2cm }+ 2L \|\bxtm - \one \hxstartm\| + \sqrt{N}\bUB  %
    \label{eq:yp1}
  \end{align}
  where in the second line we use Assumption~\ref{assumption:functions}, in the third one Assumption~\ref{assumption:bounded_variations}, and in the last line Eq.~\eqref{eq:mu}. Then by 
  adding and subtracting $\pm \one \avxtm$ inside $\|\bxtm - \one \hxstartm\|$
  and using the triangle inequality:
\begin{align}
  & \|\bdt - \one \avdt \| \nonumber\\
  &\leq\sigma_W \|\bdtm -\one\avdtm \|+L\| W\bxtm-\alpha\bdtm-\bxtm\| \nonumber
  \\ 
  & + 2 L \|\bxtm - \one \avxtm\| \!+\! 2L \sqrt{N} \|\avxtm - \hxstartm\| + \sqrt{N}\bUB.\label{eq:yp1bis}
\end{align} 
  
Consider now the second term in the right-hand side of~\eqref{eq:yp1}. We have that
  \begin{align}
    &\| W\bxt-\alpha\bdt-\bxt\|\nonumber\\
    &=\|(W-I)(\bxt - \one \avxt )-\alpha\bdt\|\nonumber\\
        &\leq \|W-I\|\|\bxt-\one\avxt \|+\alpha \|\bdt\| \nonumber\\
        &= \|W-I\|\|\bxt-\one \avxt \|+\alpha \|\bdt \pm \one \avdt \| \nonumber\\
        &\leq \|W-I\|\|\bxt-\one\avxt\|+\alpha \|\bdt-\one\avdt\| + \alpha \|\one\avdt\|\label{eq:yp2}
  \end{align}
  Moreover, regarding the term $\|\avdt\one\|$ in~\eqref{eq:yp2}, one has
  \begin{align}
    \|\one \avdt\|
        &=N\|\avdt\|=N\left\Vert \frac{1}{N}\sum_{i=1}^N \nabla f_i(\xit; t)\right\Vert\nonumber\\
        &=N\left\Vert \frac{1}{N}\sum_{i=1}^N (\nabla f_i(\xit; t)-\nabla f_i( \xstart ; t))\right\Vert\nonumber\\
        &\leq L\sum_{i=1}^N\| \xit - \hxstart \|\nonumber\\
        &\leq L\sqrt{N}\| \bxt - \one \hxstart  \|\nonumber\\
        &\leq L\sqrt{N} \|\bxt - \one \avxt \| + L\sqrt{N} \|\avxt-\hxstart \|\label{eq:yp3}
  \end{align}
  where in the second line we used the optimality of $\hxstart$, in the third one the 
  Lipschitz continuity of $\nabla \hat{f}_i$, in the fourth one the algebraic property of the 2-norm 
  and in the last one we added and subtracted $\one\avxt$ and used the triangle inequality.
  The proof is completed by combining~\eqref{eq:yp1bis}, \eqref{eq:yp2} and \eqref{eq:yp3}.
\subsection{Proof of Lemma~\ref{lemma:assumption}}

We start by considering a scalar least squares case, and then looking at the multi-dimensional scenario. 

Consider a scalar case, in which the parameter $p$ is estimated as $\hat{p}_t = p + \frac{1}{t}\sum_{i=1}^t \epsilon_i$, for measurements $m_i = p + \epsilon_i$, and $\epsilon_i \sim \mathcal{N}(0, \sigma^2)$. Then, $\hat{p}_t$ is a random variable drawn from $\mathcal{N}(p, \sigma^2/t)$. Suppose that $L \ge p > 0$ and $\mu>1$. The probability that $\hat{p}_t$ is \emph{not} in $[0, \mu L]$ is now, 
\begin{align*}%
\delta_t = \frac{1}{\sqrt{2\pi}\sigma_t} \left[\int_{-\infty}^{0} \mathrm{e}^{-\frac{1}{2}\left(\frac{x-p}{\sigma_t}\right)^2} \mathrm{d} x + \int_{\mu L}^{+\infty}  \mathrm{e}^{-\frac{1}{2}\left(\frac{x-p}{\sigma_t}\right)^2} \mathrm{d} x \right]
\end{align*}
where we have defined $\sigma_t = \sigma/\sqrt{t}$. Now, by the definition of complementary error function $\mathrm{erfc}$, the above is equivalent to
\begin{align}\notag
\delta_t = \frac{1}{2}\mathrm{erfc}\left(\frac{2 p}{\sqrt{2} \sigma_t} \right) + \frac{1}{2}\mathrm{erfc}\left(\frac{\mu L - p}{\sqrt{2} \sigma_t} \right),
\end{align}
while by the fact that $\mathrm{erfc}(x) \le \mathrm{e}^{-x^2}$ for $x>0$~\cite{Chiani2003}, then,
\begin{align}\notag
\delta_t \leq \frac{1}{2}\mathrm{e}^{-\left(\frac{2 p}{\sqrt{2} \sigma_t} \right)^2} + \frac{1}{2}\mathrm{e}^{-\left(\frac{\mu L - p}{\sqrt{2} \sigma_t} \right)^2}.
\end{align}

By the union bound, the probability that $\hat{p}_t$ is \emph{not} in $[0, \mu L]$ for \emph{at least one time} $T>t\ge\bar{t}$ is then 
\begin{align*}
  & \sum_{t=\bar{t}}^T \frac{1}{2}\mathrm{e}^{-\left(\frac{2p}{\sqrt{2} \sigma_t} \right)^2} 
  + 
  \frac{1}{2}  
  \mathrm{e}^{-\left(\frac{\mu L - p}{\sqrt{2} \sigma_t} \right)^2} 
  \\
  & \le 
  \frac{1}{2} \int_{t=\bar{t}}^T \mathrm{e}^{-\left(\frac{2p}{\sqrt{2} \sigma} \right)^2 t} 
  + \mathrm{e}^{-\left(\frac{\mu L - p}{\sqrt{2} \sigma} \right)^2 t} \mathrm{d} t 
  \\ 
  & \le \sigma^2  \left(\frac{1}{4p^2}\mathrm{e}^{-\left(\frac{2p}{\sqrt{2} \sigma} \right)^2 \bar{t}} 
  + \frac{1}{(\mu L - p)^2} \mathrm{e}^{-\left(\frac{\mu L - p}{\sqrt{2} \sigma} \right)^2 \bar{t}} \right) 
  =: \delta.
\end{align*}
Then, the probability that $\hat{p}_t$ is  in $[0, \mu L]$ for at all times $T>t>\bar{t}$, with $T \to \infty$ is $1-\delta$. Assumption~\ref{assumption:bounded_variations}(i) is then verified, since for any $\delta\in(0,1]$, one can find a finite and big enough $\bar{t}$ for which the probability that $\hat{p}_t$ is  in $[0, \mu L]$ for at all times $t\ge\bar{t}$ is $1-\delta$.

Let us look now at the multi-dimensional case. First, symmetry of $\hat{P}_{i,t}$ is guaranteed by construction. Then the fact that $\mu L_i I \ge \hat{P}_{i,t} \ge 0$ in probability, can be guaranteed looking at the asymptotical distribution of the LS estimator (Lemma~\ref{lemma:LS}), and for continuity for large enough $\bar{t}$. 

Let us start from the fact that 
  \begin{align*}
    \sqrt{t}(\hat{\xi}_{i,t} - \xi_{i,\star} ) \xrightarrow{D} \mathcal{N}(0,\Sigma_{xx}^{-1}S\Sigma_{xx}^{-1} ),
  \end{align*}
which means that
$$
\hat{\xi}_{i,t} - \xi_{i,\star}  \xrightarrow{D} \mathcal{N}(0,\Sigma_{xx}^{-1}S\Sigma_{xx}^{-1}/t )
$$
and let us analyze the asymptotic distribution. Part of $\hat{\xi}_{i,t}$ will be packed into the matrix $\hat{P}_{i,t}$, and in particular, we want to look at the probabilities that 
$
\Pr(\mu L_i I < \hat{P}_{i,t})
$
or
$
\Pr(\hat{P}_{i,t} < 0) 
$
for a $t\geq \bar{t}$, in a matrix sense. We can use the implications,
\begin{align*}
\Pr(\hat{P}_{i,t}> \mu L_i I) &\leq \Pr(\hat{P}_{i,t} - {P}_{i}> (\mu-1) L_i I),\\ \Pr(\hat{P}_{i,t} < 0 ) &\leq \Pr(\hat{P}_{i,t} - {P}_{i} < -m_i I ),
\end{align*}
and consider the latter ones. 

For sufficiently big $\bar{t}$, each entries of error matrix $\hat{P}_{i,t} - {P}_{i}$, say $e_{jk}$ is a random variable, normally distributed, with zero mean. Their covariance matrix is correlated across all entries, but we can look at an upper bound (via eigenvalue decomposition), so that $e_{jk} \sim \mathcal{N}(0,O(\max \lambda(\Sigma_{xx}^{-1}S\Sigma_{xx}^{-1}))/t )$, where $O(\cdot)$ hides a possible constant independent of $t$ and $\lambda(\cdot)$ represents the eigenvalues.

We look now at the eigenvalues of $E_{i,t} = \hat{P}_{i,t} - {P}_{i}$, say $\lambda(E_{i,t})$. We have the following line of implications,
\begin{align*}
\Pr(\hat{P}_{i,t} \!-\! {P}_{i}> (\mu\!-\!1) L_i I) &= \Pr(\min_j \,\lambda_j(E_{i,t}) >(\mu\!-\!1) L_i)\\ 
& \hskip-15mm\leq \Pr(\frac{1}{n}\sum_j \,\lambda_j(E_{i,t}) >(\mu-1) L_i)\\
& \hskip-15mm\leq \Pr(\frac{1}{n}\sum_j {e_{jj} + \sum_{k, k\neq j}|e_{jk}|} >(\mu-1) L_i)\\
& \hskip-15mm\leq \Pr( n \max_{j,k}|e_{jk}| >(\mu-1) L_i) \\
& \hskip-15mm\leq n^2\Pr( n |e_{jk}| >(\mu-1) L_i) \\
& \hskip-15mm= 2 n^2\Pr( n e_{jk} >(\mu-1) L_i).
\end{align*}
Where the first line comes from the definition of definite positiveness; the second line comes from the definition of min; the third line from Gershgorin eigenvalue circle theorem; the fourth line from the definition of the max, and the last two lines from the definition of probability and its symmetric property w.r.t. the zero mean value. 

So the $\Pr(\hat{P}_{i,t}> \mu L_i I)$ is less or equal than  $2 n^2\Pr(e_{jk} >(\mu-1) L_i/n)$, for normal random variable $e_{jk}\sim \mathcal{N}(0,O(\max \lambda(\Sigma_{xx}^{-1}S\Sigma_{xx}^{-1}))/t )$. A similar expression holds for $\Pr(\hat{P}_{i,t} < 0 )$ which is less or equal than $2 n^2\Pr( e_{jk} >m_i /n)$. Therefore the probability that $\hat{P}_{i,t}$ is not in $[0, \mu L_i]$ for at least one time $T>t\geq \bar{t}$ can be computed as 
$$
\sum_{t = \bar{t}}^T 2 n^2\Pr( e_{jk} >\min \{(\mu-1) L_i/n,m_i /n\} ) =: \delta.
$$
The existence of a finite and computable $\delta$ can be proven in the same way as the scalar case (since now we are in a scalar case) given that the variance scales as $1/{t}$, which proves the claim. 

\subsection{Proof of Lemma~\ref{lemma:RLS}}

The inequality follows from Lemma~\ref{lemma:LS}, 
\begin{align*}
\left \vert \hUit(x) - U_i(x) \right \vert
  =
\left \vert  (\xi_i^\star - \hait )^\top \tv \right \vert
 \le
\left \Vert \xi_i^\star - \hait \right \Vert \left \Vert \tv \right \Vert,
\end{align*}
where we set $\tv\triangleq\col(1, x, [x]_1x/2, \dots, [x]_Nx/2)$. 
Since by hypothesis $\|x\|$ is bounded, and by Lemma~\ref{lemma:LS} $\|\xi_i^\star - \hait\| \leq O(1/\sqrt{t})$, then, $|\hUit(x) - U_i(x)  | \leq O(1/\sqrt{t})$.

\subsection{Proof of agent $j$ specific regret bound}
\label{extra-bound}
Consider the agent $j$ specific regret
\begin{equation*}
\mathcal{R}_{j,T} \triangleq R_{T}(\{x_{j,t}\}_{t=1}^T) = \sum_{t=1}^T \left (\sum_{i=1}^N f_i(x_{j,t}; t)  - \fstart \right )
\end{equation*}
We can bound similarly to how we bounded $R_T(\{\avxt\}_{t=1}^T)$ in the proof of Theorem~\ref{thm:regret}.

First, note that since the consensus metric $C_T$ is bounded and $\avxt$ is bounded (both from Theorem~\ref{thm:regret}), then $x_{j,t}$ is bounded. Therefore, we can use the same steps of the proof of Theorem~\ref{thm:regret} applied to $x_{j,t}$ instead of $\avxt$ up to Equation~\eqref{eq:RT1}, which reads now
\begin{align}%
  R_T(\{x_{j,t}\}_{t=1}^T)
  \le 
  O(\bar{t}) \!+\! O(c_U) \!+\! \sum_{t=\bar{t}}^T\! \Big (\! \sum_{i=1}^N \hf_i( x_{j,t}; t) - \hf_\star(t)\!\Big ).
\label{eq:RTX}
\end{align}
We bound now the second term with the following line of reasoning, (with $\hf(x;t) = \sum_{i=1}^N\hf_i(x;t)$)
\begin{align*}
\hf(x_{j,t}; t) - \hf(\avxt; t) &\leq \nabla \hf(\avxt; t)^\top (x_{j,t}-\avxt) + \frac{L}{2}\|x_{j,t}-\avxt\|^2\\
&\hspace*{-2.5cm}\leq  (\nabla \hf(\avxt; t)-\nabla \hf(\hat{x}_{\star}(t); t))^\top (x_{j,t}-\avxt) + \frac{L}{2}\|x_{j,t}-\avxt\|^2\\
&\hspace*{-2.5cm}\leq L\|\avxt-\hat{x}_{\star}(t)\|\|x_{j,t}-\avxt\| + \frac{L}{2}\|x_{j,t}-\avxt\|^2 \\
&\hspace*{-2.5cm}\leq L\|\Vt\|\|\Vt\| + \frac{L}{2}\|\Vt\|^2  = \frac{3}{2} L\|\Vt\|^2.
\end{align*}
Therefore, $\sum_{i=1}^N \hf_i( x_{j,t}; t) - \hf_\star(t) \leq \sum_{i=1}^N \hf_i( \avxt; t) - \hf_\star(t) + \frac{3}{2} L\|\Vt\|^2 $, and by using Theorem~\ref{thm:GT_convergence}:
\begin{multline}
  \sum_{t=\bar{t}}^T \left ( \sum_{i=1}^N \hf_i( x_{j,t}; t) - \hf_\star(t) \right )
  \\  
  \le
  \sum_{t=\bar{t}}^T O(\rho^{t-\bar{t}}) + 
  O\left( (T-\bar{t})\frac{2 L (N \bUB^2 + \cUB^2)}{ (1-\rho)^2}\right),\label{eq:hf_boundX}
\end{multline}
from which
\begin{align*}
  R_T(\{x_{j,t}\}_{t=1}^T)
  & 
  \le O(1) + O (c_U) + O\left( T\, \frac{2 L (N \bUB^2 + \cUB^2)}{(1-\rho)^2}\right), 
\end{align*}
and the thesis is proven.

\end{document}